\documentclass[]{interact}

\usepackage{epstopdf}% To incorporate .eps illustrations using PDFLaTeX, etc.
\usepackage[caption=false]{subfig}% Support for small, `sub' figures and tables

\usepackage{tikz}
\usetikzlibrary{decorations.pathreplacing,shapes}
\tikzset{pblock/.style = {draw,rectangle split, rectangle split horizontal,
                      rectangle split parts=2, align=center}}

\usepackage[numbers,sort&compress]{natbib}% Citation support using natbib.sty
\bibpunct[, ]{[}{]}{,}{n}{,}{,}% Citation support using natbib.sty
\theoremstyle{plain}% Theorem-like structures provided by amsthm.sty
\newtheorem{theorem}{Theorem}[section]
\newtheorem{lemma}[theorem]{Lemma}

\newtheorem{proposition}[theorem]{Proposition}

\theoremstyle{definition}
\newtheorem{definition}[theorem]{Definition}

\theoremstyle{remark}

\newtheorem{notation}{Notation}

\newcommand{\dist}{\ensuremath\mathrm{dist}}
\newcommand{\Zplus}{\ensuremath{\mathbb Z}_{\geqslant0}}
\newcommand{\overnodes}[1]{\ensuremath L_V(#1)}
\newcommand{\overedges}[1]{\ensuremath L_E(#1)}

\newcommand{\indentitem}[1]{{\setlength\itemindent{1.5em} \item #1}}

\newenvironment{procedure}[2]{
\begin{center}
\begin{tabular}{|p{0.95\textwidth}|}
\hline
\multicolumn{1}{|c|}{\bfseries #1}\\
\hline
#2

\small
\begin{enumerate}
}{\end{enumerate}\\
\hline
\end{tabular}
\end{center}
}

\begin{document}

%\articletype{ARTICLE TEMPLATE}% Specify the article type or omit as appropriate

\title{Two-machine routing open shop on a tree: instance reduction and efficiently solvable subclass}

\author{
\name{I.~D. Chernykh\textsuperscript{a,b,c}\thanks{CONTACT I.~D. Chernykh. Email: idchern@math.nsc.ru} and E.~V. Lgotina\textsuperscript{b}}
\affil{\textsuperscript{a}Sobolev Institute of Mathematics, Novosibirsk, Russia; \textsuperscript{b}Novosibirsk State University, Novosibirsk, Russia; \textsuperscript{c}Novosibirsk Technical State University, Novosibirsk, Russia}
}

\maketitle

\begin{abstract}
%The routing open shop problem is a natural combination of the metric traveling salesman problem and the classical open shop scheduling problem. Both counterparts generally are NP-hard, being however polynomially solvable in special cases considered: TSP is trivial on a tree, and open shop is solvable in linear time in case of two machines by a well-known Gonzalez-Sahni algorithm (while being NP-hard for three and more machines). Surprisingly, the combination of those problem becomes NP-hard even in a simplest case of a two-machine routing open shop on a link. 

%We describe an instance reduction procedure for the two-machine problem on arbitrary tree. The procedure preserves the standard lower bound on the makespan and allows to describe wide polynomially solvable subclasses of the problem. Conditions, describing these classes, allow building an optimal schedule in linear time. For any instance from those classes optimal makespan coincides with the standard lower bound. 
We consider two-machine routing open shop problem on a tree. In this problem a transportation network with a tree-like structure is given, and each node contains some jobs to be processed by two mobile machines. Machines are initially located in the predefined node called \emph{the depot}, have to traverse the network to perform their operations on each job (and each job has to be processed by both machines in arbitrary order), and machines have to return to the depot after performing all the operations. The goal is to construct a feasible schedule for machines to process all the jobs and to return to the depot in shortest time possible. This problem is known to be NP-hard even in the case when the transportation network consists of just two nodes.

We propose an instance reduction procedure which allows to transform any instance of the problem to a simplified instance on a chain with limited number of jobs. The reduction considered preserves the standard lower bound on the optimum. We describe four possible outcomes of this procedure and prove that in three of them the initial instance can be solved to the optimum in linear time, thus introducing a wide polynomially solvable subclass of the problem considered. Our research can be used as a foundation to construct efficient approximation algorithms for the two-machine routing open shop on a tree.
\end{abstract}

\begin{keywords}
Scheduling; open shop with delays; routing open shop; standard lower bound; instance reduction; polynomially solvable subclass; overloaded node; overloaded edge
\end{keywords}

\section{Introduction}
We consider the routing open shop problem, which is a natural combination of a well-known metric traveling salesman problem and a classical scheduling open shop problem. Metric TSP hardly needs an introduction. The open shop problem, introduced by Gonzalez and Sahni \cite{GoSa}, can be described as follows. Sets \({\cal M}=\{M_1,\ldots,M_m\}\) of machines and \({\cal J}=\{J_1,\ldots,J_n\}\) of jobs are given and each machine \(M_i\) has to perform an operation on each job \(J_j\), this operation \(O_{ji}\) requires \(p_{ji}\geqslant0\) time units to complete. Each machine has to process jobs in some sequence which is not given in advance but has to be chosen by a scheduler. Operations of the same job cannot be processed simultaneously. The goal is to construct a schedule of processing of all jobs, {\it i.e.} to specify non-negative starting and completion times for each operation, such that the conditions above are satisfied and the maximum completion time (also referred to as the \emph{makespan}) is minimized.
Following the standard three-field notation for scheduling problems (see \cite{TheBible} for example) the open shop problem with \(m\) machines is denoted by \(Om||C_{\max}\). Notation \(O||C_{\max}\) is used when the number of machines is not bounded by any constant. The \(Om||C_{\max}\) problem is known (\cite{GoSa}) to be polynomially solvable in the case of two machines and is NP-hard for \(m\geqslant 3\). The  \(O||C_{\max}\) problem 
 is strongly NP-hard. Moreover, unless \(P=NP\), no \(\rho\)-approximation algorithm for \(O||C_{\max}\) exists \(\rho<\frac54\) \cite{SSS}.

Several algorithms for solving the two-machine problem $O2||C_{\max}$ to the optimum were proposed over the last decades, for example by Gonzalez and Sahni \cite{GoSa}, Pinedo and Schrage \cite{PinSch82} and de Werra \cite{dW89}. All the known algorithms run in linear time and produce optimal schedules with different structures. An important property of the schedules produced by each of those algorithms (and therefore of the optimal schedule) is its so-called \emph{normality}: the makespan of those schedules always coincides with the \emph{standard lower bound} \(\bar C\doteq\max\left\{\max\limits_i\sum\limits_{j=1}^np_{ji},\max\limits_j\sum\limits_{i=1}^mp_{ji}\right\}\).
%Let us  describe the properties of schedule $S_{{\mathrm GS}}$ constructed by the Gonzalez-Sahni algorithm. First, operations of all jobs from ${\cal J}$ except for one (so-called \emph{diagonal} job) are performed in $S_{{\mathrm GS}}$ in the same order: either $M_1\to M_2$ or $M_2\to M_1$ depending on the properties of the diagonal job. Second, the makespan of $S_{{\mathrm GS}}$ always coincides with the \emph{standard lower bound} \(\bar C\doteq\max\left\{\max\limits_i\sum\limits_{j=1}^np_{ji},\max\limits_j\sum\limits_{i=1}^mp_{ji}\right\}\). %Following definitions from \cite{KoSeCh99} (see Section \ref{sec:prelim}) the last property of a schedule will be referred to as the \emph{normality}.

\medskip

Most of the classical scheduling models (open shop included) share the following disadvantage. It is supposed that each machine is able to start a new operation at the same moment when it completes the previous one. In a real life environment that's not always possible. Usually jobs represent some material objects, therefore some delays between processing operations of two subsequent jobs may be unavoidable. Such delays can be machine-dependent, job- or sequence-dependent, and taking them into account can make the problem harder to investigate. Still, there is a number of papers considering scheduling problems with \emph{transportation delays} (see \cite{BrKn,LuSoStru,Stru,Kl93} for example). However, the problem we are considering in this paper uses a different approach to model transportation delays.

We consider the \emph{routing open shop} problem \cite{AvBeCh1} which can be described as the open shop meeting the metric traveling salesman problem (TSP). Let the input of the TSP be given by an edge-weighted graph \(G\). Jobs from \({\cal J}\) are distributed between the nodes of \(G\), each node contains at least one job. Machines are mobile and are initially located at the predefined node referred to as \emph{the depot}. Machines have to travel over the edges of \(G\), weights of the edges represent travel times for each machine. % and satisfy the triangle inequality. 
Any number of machines can travel over the same edge at the same time. Each machine has to visit each node of \(G\) (not necessary once), perform all the respective operations (under the feasibility constraints from the open shop problem), and return back to the depot after processing all the jobs. One has to construct a schedule, specifying starting and completion times for each machine's \emph{activity}, which is either a performing of some operation or traveling over some edge of $G$. Note that a machine has to arrive at some node to be able to process a job located at that node.

The makespan \(R_{\max}\) of a schedule $S$ for the routing open shop is the maximum completion time of machine's activity. (Note that $R_{\max}\geqslant C_{\max}$ and those to values are different in case the last activity is a traveling of machine to the depot.) The goal is to construct a feasible schedule minimizing the makespan.
The routing open shop problem with \(m\) machines is denoted by \(ROm||R_{\max}\), or \(ROm|G=X|R_{\max}\) if we want to specify the structure \(X\) of the graph \(G\). In the latter case we use either standard notation from graph theory, such as $K_p$ for the \emph{complete graph} with $p$ nodes, or standard terms like $tree$ or $chain$.

The $ROm||R_{\max}$ problem has a certain similarity to the so-called {\em open shop with batch setup times} (see \cite{Kl93} for example). In the latter problem jobs a partitioned into several groups referred to as \emph{batches}, and a machine has to spend a pre-defined {\em setup time} when switching from one batch to another. Batches can correspond to sets of jobs from the same node in the routing open shop, and setup times correspond to travel times. However, there are two significant differences. First, setup times usually only depend on the destination batch (in terms of the routing scheduling problem that would mean that the travel time between nodes $u$ and $v$ depend only on $v$); second, there is no initial state of machines which can correspond to the depot. Including the depot into the picture makes the problem to be a generalization of a metric TSP, and it worth mentioning that such combinations of hard discrete optimization problems get more attention during the last decades. For example, the routing scheduling model appeared independently while considering tasks arising both in production (see, e.g. \cite{AvBe96,AvBe99}), so in the service industry \cite{ChouLin07,YuLinChou10}.

The general routing open shop problem contains the metric TSP as a special case, moreover, the problem with a single machine is equivalent to the metric TSP and therefore is strongly NP-hard. On the other hand, the problem with zero travel times (or  with \(G=K_1\)) is equivalent to the open shop problem and is NP-hard for \(m\geqslant3\). However it is known, that the routing open shop problem remains NP-hard even in the simplest case \(RO2|G=K_2|R_{\max}\) with two machines and just two nodes of the transportation network \cite{AvBeCh1}. On the other hand, FPTAS for such a case is described in \cite{KoFPTAS}.

Our research aims on the description of wide polynomially solvable cases of NP-hard problem $RO2||R_{\max}$. A few such cases for the $RO2|G=K_2|R_{\max}$ problem can be found in \cite{KoFPTAS,ChPya}, see Sections \ref{sec:prelim}, \ref{sec:reduction} and \ref{sec:chaincases} for details.  

%An FPTAS and a couple of polynomially solvable cases for \(RO2|G=K_2|R_{\max}\) are presented in \cite{KoFPTAS}. %Those subcases are formulated in terms of the properties of the diagonal job and guarantee the normality of the optimal schedule (i.e. the optimal makespan coincides with the standard lower bound). This is described in detail in Section \ref{sec:prelim}.Another polynomially solvable special case of $RO2|G=K_2|R_{\max}$ was introduced in \cite{ChPya}. It is based on the so-called \emph{superoverloaded node}, and also guarantee the normality of the optimal schedule. Recently \cite{ChPyaToAppear} that result was  generalized for the $RO2|easy-TSP|R_{\max}$ (i.e. the problem in which the transportation network structure or the distance matrix allows to solve the underlying TSP to the optimum in polynomial time) under some restrictions on the location of the superoverloaded node. See Sections \ref{sec:prelim}, \ref{sec:reduction} and \ref{sec:chaincases} for details. 

Although our research focuses on the two-machine version of the problem, the progress made in the study of the $ROm||R_{\max}$ problem should be mentioned as well. A series of approximation algorithms for the $ROm||R_{\max}$ problem was developed, starting with the $\frac{m+4}{2}$-approximation \cite{AvBeCh1}. The best known algorithm up to date has the approximation ration guarantee of $O(\log m)$ \cite{Kologm}. (An intriguing open question is whether an approximation algorithm with a constant approximation ratio exists for the $ROm||R_{\max}$ problem.)
A number of papers is devoted to the research of a special case with unit processing times \cite{GolPya19,BevPyat16,BevPyatSev19}.

\medskip

In this paper we consider the problem \(RO2|G=tree|R_{\max}\), and describe several special cases which are solvable to the optimum in linear time, with the optimal makespan equal to the standard lower bound. The main special cases are formulated in terms of load distribution between the nodes, the formulation involves the definitions of an overloaded node and an \emph{overloaded edge} (see Section \ref{sec:prelim}), and is based on a special procedure of instance reduction and its properties (Section \ref{sec:reduction}). For the sake of completeness we also provide a couple of additional special cases for $RO2|G=chain|R_{\max}$ problem formulated in terms of properties of the diagonal job, which plays important role in the Gonzalez-Sahni algorithm for the $O2||C_{\max}$ problem (Theorem \ref{theo:kononovextended}). These cases are elementary extensions of known classes for $RO2|G=K_2|R_{\max}$ described by Kononov in \cite{KoFPTAS} (see Theorem \ref{theo:kononovcases}).

The structure of the remainder of the paper is as follows. Section \ref{sec:prelim} contains a detailed problem description, necessary notation and the formulation of known results we use. In Section \ref{sec:reduction}, we describe the procedure of instance reduction, which is the main part of our algorithm. Polynomially solvable outcomes of the instance reduction procedure are described in Section \ref{sec:chaincases}, followed by the description of sufficient conditions of polynomial solvability in terms of the properties of the initial instance in Section \ref{sec:conditions}. Concluding remarks and some open questions are given in Section \ref{sec:conclusion}.

\section{Preliminary notes}\label{sec:prelim}
Let us give a formal description of the routing open shop problem.

A problem instance combines inputs from the metric TSP and the open shop problem in the following manner.
A connected graph \(G=\langle V,E\rangle\) is given, a non-negative weight function \(\tau:E\to\Zplus\) is defined. 
%The function \(\tau\) satisfies the triangle inequality. 
One of the nodes \(v_0\in V\) is chosen to be the \emph{depot}. 
Jobs from the set \({\cal J}=\{J_1,\ldots,J_n\}\) are distributed among the nodes from \(V\). A set of jobs located at \(v\in V\) is denoted by \({\cal J}(v)\) and is non-empty for any node with possible exclusion of the depot. Machines from the given set \({\cal M}=\{M_1,\ldots,M_m\}\) are initially located at the depot and each machine can travel over the edges of \(G\), travel time of each machine over an edge \(e\in E\) equal to \(\tau(e)\). Any number of machines can travel over the same edge in any direction at the same time. Machines are allowed to visit each node multiple times therefore we assume machines use the shortest paths while traveling from one location to another.  
Each machine \(M_i\) has to perform an \emph{operation} \(O_{ji}\) on every job \(J_j\). This operation takes \(p_{ji}\in\Zplus\) time units and requires the machine to be at the location of \(J_j\): while machine is in the node \(v\), it can only process operations of jobs from \({\cal J}(v)\). Different operations of the same job cannot be processed simultaneously, and each machine can process at most one operation at a time. Machines have to return to the depot after processing all the operations. We use notation \(p_{ji}(I)\), \(G(I)\), \(\tau(I;e)\) and \({\cal J}(I;v)\), if we want to specify a problem instance \(I\). 

A \emph{schedule} \(S\) can be described by specifying the starting time $s_{ji}$ for each operation $O_{ji}$:
\[S=\left\{s_{ji}|i=1,\ldots,m,\,j=1,\ldots,n\right\}.\]
The completion time of operation $O_{ji}$ in a schedule $S$ is denoted by $c_{ji}(S)=s_{ji}(S)+p_{ji}$, notation $S$ is omitted when not needed.

Let \(\dist(v,u)\) denote the weighted distance between the nodes \(v\) and \(u\) (and {\em vice versa}), {\it i.e.} the minimal total weight of edges belonging to some chain connecting \(v\) and \(u\). So \(\dist(v,u)\) is the shortest time needed for a machine to reach \(u\) from \(v\). We also use notation \(\dist(I;v,u)\) for a specific instance \(I\).

\begin{definition}A schedule \(S\) for an instance \(I\) is referred to as \emph{feasible} if it satisfies the following conditions:
\begin{enumerate}
	\item If $i_1=i_2$ or $j_1=j_2$ (but not both) then
	\[(s_{j_1i_1},c_{j_1i_1})\cap(s_{j_2i_2},c_{j_2i_2})=\emptyset.\]
	\item If operation of job \(J_j\in{\cal J}(v)\) is the first to start by machine \(M_i\) then
	\[s_{ji}\geqslant\dist(I;v_0,v).\]
	\item If machine \(M_i\) processes operation \(O_{ji}\) before the processing of an operation \(O_{j'i}\), \(J_j\in{\cal J}(v)\), and \(J_{j'}\in{\cal J}(v')\), then
	\[s_{j'i}\geqslant c_{ji}+\dist(I;v,v').\]
\end{enumerate}
\end{definition}
Condition 1 means that intervals of processing of \emph{dependent} operations ({\it i.e.} operations of the same job or of the same machine) do not overlap. Conditions 2 and 3 mean that machine cannot start an operation before it reaches its location.

Suppose an operation of job \(J_j\in{\cal J}(v)\) is the last to be processed by machine \(M_i\) in some schedule \(S\). Then we define the \emph{release time} of machine \(M_i\) as
\[R_i(S)\doteq c_{ji}(S)+\dist(v,v_0).\]
The \emph{makespan} of schedule \(S\) is \(R_{\max}(S)\doteq\max\limits_iR_i(S)\). The goal is to find a feasible schedule minimizing the makespan.

\medskip

For some problem instance \(I\) we use the following
\begin{notation}
\begin{itemize}
	\item \(\ell_i(I)\doteq\sum\limits_{j=1}^np_{ji}(I)\) --- the \emph{load} of machine \(M_i\);
	\item \(\ell_{\max}(I)\doteq\max\limits_i\ell_i(I)\) --- the maximal machine load;
	\item \(d_j(I)\doteq\sum\limits_{i=1}^mp_{ji}(I)\) --- the \emph{length} of job \(J_j\);
	\item \(d_{\max}(I;v)\doteq\max\limits_{J_j\in{\cal J}(v)}d_j(I)\) --- the maximal job length at node \(v\);
	\item \(\Delta(I;v)\doteq\sum\limits_{J_j\in{\cal J}(v)}d_j(I)\) --- the \emph{load} of node \(v\);
	\item \(\Delta(I)\doteq\sum\limits_{v\in V}\Delta(I;v)\) --- the \emph{total load} of instance \(I\);
	\item \(T^*(I)\) --- the optimum of the underlying TSP, i.e. the length of the shortest a cyclic route visiting each node at least once;
	\item \(R^*_{\max}(I)\) --- the optimal makespan.
\end{itemize}
We omit \(I\) from the notation in case when it does not lead to a confusion.
\end{notation}

The following \emph{standard lower bound} on the optimum for the routing open shop problem was introduced in \cite{AvBeCh2}:
\begin{equation}\label{eq:SLB}\bar R(I)\doteq\max\left\{\ell_{\max}(I)+T^*(I),\max_{v\in V}\bigl(d_{\max}(I;v)+2\dist(I;v_0,v)\bigr)\right\}.\end{equation}
Note that \(\bar R\) coincides with \(\bar C\) in case when all edges have zero weight or $G=K_1$ (in this case our problem is reduced to the classical open shop problem).

Our study is focused on the case of two machines. In this case we use simplified notation for the operations of each job \(J_j\): \(a_j\) and \(b_j\) instead of \(O_{j1}\) and \(O_{j2}\), respectively. Moreover, we use the same notation (\(a_j\) and \(b_j\)) for operations' processing times whenever it does not lead to a confusion.

\medskip

We use the following definitions inherited from \cite{KoSeCh99}.

\begin{definition} 
A feasible schedule $S$ for a problem instance $I$ is referred to as \emph{normal} if $R_{\max}(S)=\bar R(I)$. Instance $I$ is \emph{normal} if it admits construction of a normal schedule.

A class of instances is \emph{normal} if it consists of normal instances only. A normal class ${\cal K}$ is referred to as \emph{efficiently normal} if there exists a polynomial time algorithm for solving any instance from ${\cal K}$ to the optimum.
\end{definition}

The goal of this paper is to describe wide efficiently normal classes for the $RO2|G=tree|R_{\max}$ problem. 
Below we describe a few such classes known from previous research.

The first efficiently normal class of instances of $RO2|G=K_2|R_{\max}$ is due to Kononov \cite{KoFPTAS} and its description is based on so-called diagonal job, which can be defined as follows.

\begin{definition}The \emph{diagonal} job of an instance of the $O2||C_{\max}$ (or $RO2||R_{\max}$) problem is such a job $J_r\in{\cal J}$ that
\[r=\arg\max_j\bigl\{\min\{a_j,b_j\}\bigr\}.\]
\end{definition}

We also need the following

\begin{definition}\label{def:earlyschedule}A feasible schedule $S$ if referred to as \emph{early} if no operation $O_{ji}$ can start earlier than at $s_{ji}(S)$, providing that the sequences of operations of each job and each machine from $S$ are preserved, without violating the feasibility. 
\end{definition}
Note that any early schedule is uniquely defined by sequences of operations of each job and each machine.

\begin{theorem}[Kononov, \cite{KoFPTAS}]\label{theo:kononovcases}
A class of instances of the \(RO2|G=K_2|R_{\max}\) problem satisfying at least one of the following properties of the diagonal job $J_r$ is solvable to the optimum in linear time:
\begin{enumerate}
	\item $J_r\in{\cal J}(v_0)$,
	\item $d_r\geqslant\ell_{\max}$.
\end{enumerate}
\end{theorem}

The proof of Theorem \ref{theo:kononovcases} is based on properties of the Gonzalez-Sahni algorithm for open shop problem $O2||C_{\max}$ \cite{GoSa}. This proof can be easily extended on a problem $RO2|G=chain|R_{\max}$ under the following conditions:
\begin{enumerate}
	\item the depot $v_0$ is one of the terminal nodes of chain $G$,
	\item job $J_r$ is located at some terminal node of $G$.
\end{enumerate} To present this proof, we use the following Gonzalez-Sahni formulation, similar to that described in \cite{KoFPTAS}.

\begin{procedure}{Gonzalez-Sahni algorithm \cite{GoSa}}
{{\tt Input:} An instance $I$ of $O2||C_{\max}$ problem.\\
{\tt Output:} An optimal schedule for $I$.}
\indentitem{Partition the set of jobs ${\cal J}$ into two subsets: 
\[{\cal J}_A=\{J_j|a_j\leqslant b_j\}\text{, }{\cal J}_B=\{J_j|a_j> b_j\}.\]}
\indentitem{Let $J_r$ be a diagonal job. Without loss of generality assume $J_r\in{\cal J}_A$.}
\indentitem{Choose enumerations of operations from ${\cal J}\setminus\{J_r\}$ ($\mathfrak A$ and $\mathfrak B$  for machines $M_1$, $M_2$ respectively) in the following way:
\begin{enumerate}
	\item {\bfseries If} $d_r<\ell_{\max}$ {\bfseries then}  $\mathfrak A=\mathfrak B$ is a concatenation of arbitrary enumerations of ${\cal J}_A\setminus\{J_r\}$ and ${\cal J}_B$,
	\item {\bfseries If} $d_r\geqslant\ell_{\max}$ {\bfseries then} both $\mathfrak A$ and $\mathfrak B$ are arbitrary and independent.
\end{enumerate} }
\indentitem{Construct an early schedule in the following manner:
\begin{enumerate}
	\item Machine $M_1$ processes operation of jobs from ${\cal J}\setminus\{J_r\}$ according to $\mathfrak A$, then $a_r$;
	\item Machine $M_2$ processes $b_r$, then operations of jobs from ${\cal J}\setminus\{J_r\}$ according to $\mathfrak B$;
	\item Operations of all jobs except $J_r$ are processed first by $M_1$, then by $M_2$.
\end{enumerate}}
\end{procedure}

\begin{lemma}[\cite{GoSa}]\label{lem:GoSa} The Gonzalez-Sahni algorithm runs in $O(n)$ time and obtains a normal schedule for any instance of $O2||C_{\max}$.\end{lemma}

\begin{theorem}\label{theo:kononovextended}
Let $I$ be an instance of the \(RO2|G=chain|R_{\max}\) problem with $J_r$ being a diagonal job, $G$ is a chain $(v_0,\dots,v_g)$. Then any of the following conditions implies $I$ is normal and an optimal schedule for $I$ can be found in $O(n)$:
\begin{enumerate}
	\item $J_r\in{\cal J}(v_0)$,
	\item $J_r\in {\cal J}(v)$ and $d_r\geqslant\ell_{\max}$.
\end{enumerate}
\end{theorem}

\begin{proof}The algorithm for \(RO2|G=chain|R_{\max}\) is based on the Gonzalez-Sahni algorithm and its properties. The key fact is that operations of jobs from each of the sets ${\cal J}_A$ and ${\cal J}_B$ (Step 3) can be processed in an arbitrary order, and we show that is it possible to choose the order in such a way that each machine is guaranteed to take an optimal route. Thus the algorithm of building an optimal schedule in both cases has the following structure:
\begin{itemize}
	\item Choose a specific orders $\mathfrak A$ and $\mathfrak B$ for operations of non-diagonal jobs to apply at Step 3 of Gonzalez-Sahni algorithm.
	\item Build a schedule $S_{\mathrm GS}$ ignoring travel times using Gonzalez-Sahni algorithm.
	\item ``Insert'' travel times into $S_{\mathrm GS}$ to obtain a schedule for the initial $RO2||R_{\max}$ problem instance.
\end{itemize}

Hereafter we assume without loss of generality that $J_r\in {\cal J}_A$, and use notation ${\cal J}_A(v)\doteq{\cal J}_A\cap{\cal J}(v)\setminus\{J_r\}$ (${\cal J}_B(v)\doteq{\cal J}_B\cap{\cal J}(v)$) for each $v\in V$. Let us specify the orders $\mathfrak A$ and 
$\mathfrak B$ for both cases of the Theorem.

{\bfseries Case 1.} Let ${\mathfrak A}={\mathfrak B}$ be a concatenation of arbitrary enumerations of ${\cal J}_A(v_0),{\cal J}_A(v_1),\dots,{\cal J}_A(v_g),{\cal J}_B(v_g),\dots,{\cal J}_B(v_0)$. The structure of the resulting schedule is shown in Figure \ref{fig:KeC1}. Thick arcs represent multiple precedence constraints: not between the two blocks, but between respective operations of the same job.

{\bfseries Case 2.} Let ${\mathfrak A}$ be a concatenation of arbitrary enumerations of ${\cal J}(v_0),{\cal J}(v_1),\dots,{\cal J}(v_g)\setminus\{J_r\}$, and ${\mathfrak B}$ be a concatenation of arbitrary enumerations of ${\cal J}(v_g)\setminus\{J_r\},{\cal J}(v_{g-1}),\dots,{\cal J}(v_0)$ (see Figure \ref{fig:KeC2}). A thick dashed line represents the connection between operations of the diagonal job: $c_{r2}=s_{r1}$.   

Note that in both cases the orders $\mathfrak A$ and $\mathfrak B$ comply with the Gonzalez-Sahni algorithm, hence the schedule built is normal. Now it is sufficient to observe that inserting travel times into those schedules does not introduce extra delay intervals into their structure.
\end{proof}

\begin{figure}%
\begin{center}
	\begin{tikzpicture}[yscale=1.2]
		\node[draw] (aA0) at (0,1) {${\cal J}_A(v_0)$};
		\node[draw] (aA1) at (2,1) {${\cal J}_A(v_1)$};
		\node (Adots1) at (3.7,1) {$\cdots$};
		\node[pblock] (ag) at (6,1) {\nodepart{one}${\cal J}_A(v_g)$
																 \nodepart{two}${\cal J}_B(v_g)$};
		\node (Adots2) at (8.4,1) {$\cdots$}; 
		\node[draw] (aB0) at (10,1) {${\cal J}_B(v_0)$};
		\node[draw,circle] (ar) at (11.5,1) {$a_r$};
		\draw[->] (aA0) -- (aA1);
		\draw[->] (aA1) -- (Adots1);
		\draw[->] (Adots1) -- (ag);
		\draw[->] (ag) -- (Adots2);
		\draw[->] (Adots2) -- (aB0);
		\draw[->] (aB0) -- (ar);
		
		\node[draw,circle] (br) at (0,0) {$b_r$};
		\draw (br) edge[bend left=10] (2,0.5);
		\draw (2,0.5) -- (9.5,0.5);
		\draw (9.5,0.5) edge[->,bend right=10] (ar);
		
		\node[draw] (bA0) at (1.5,0) {${\cal J}_A(v_0)$};
		\node[draw] (bA1) at (3.5,0) {${\cal J}_A(v_1)$};
		\node (Bdots1) at (5.2,0) {$\cdots$};
		\node[pblock] (bg) at (7.5,0) {\nodepart{one}${\cal J}_A(v_g)$
																 \nodepart{two}${\cal J}_B(v_g)$};
		\node (Bdots2) at (9.9,0) {$\cdots$}; 
		\node[draw] (bB0) at (11.5,0) {${\cal J}_B(v_0)$};
		
		\draw[->] (br) -- (bA0);
		\draw[->] (bA0) -- (bA1);
		\draw[->] (bA1) -- (Bdots1);
		\draw[->] (Bdots1) -- (bg);
		\draw[->] (bg) -- (Bdots2);
		\draw[->] (Bdots2) -- (bB0);
		
		\draw[very thick,->] (aA0) -- (bA0);
		\draw[very thick,->] (aA1) -- (bA1);
		\draw[very thick,->] (ag) -- (bg);
		\draw[very thick,->] (aB0) -- (bB0);
		
	\end{tikzpicture}
\end{center}
\caption{Structure of the schedule for Theorem \ref{theo:kononovextended}, Case 1.}%
\label{fig:KeC1}%
\end{figure}

\begin{figure}%
\begin{center}
	\begin{tikzpicture}[yscale=1.2]
		\node[draw] (A0) at (-3,1) {${\cal J}(v_0)$};
		\node[draw] (A1) at (-1,1) {${\cal J}(v_1)$};
		\node (Adots) at (0.5,1) {$\cdots$};
		\node[draw] (Ag) at (2,1) {${\cal J}(v_g)$};
		\node[draw,ellipse] (ar) at (6,1) {\makebox[4cm]{$a_r$}};
		\node[draw,ellipse] (br) at (0,0) {\makebox[4cm]{$b_r$}};
		\node[draw] (Bg) at (4,0) {${\cal J}(v_g)$};
		\node (Bdots) at (5.5,0) {$\cdots$};
		\node[draw] (B1) at (7,0) {${\cal J}(v_1)$};
		\node[draw] (B0) at (9,0) {${\cal J}(v_0)$};
		
		\draw[->] (A0) -- (A1);
		\draw[->] (A1) -- (Adots);
		\draw[->] (Adots) -- (Ag);
		\draw[->] (Bg) -- (Bdots);
		\draw[->] (Bdots) -- (B1);
		\draw[->] (B1) -- (B0);
		
		\draw[->,dotted] (Ag) -- (Bg);
		\draw[very thick,dashed] (br.east) -- (ar.west);
		
	\end{tikzpicture}
\end{center}
\caption{Structure of the schedule for Theorem \ref{theo:kononovextended}, Case 2.}%
\label{fig:KeC2}%
\end{figure}
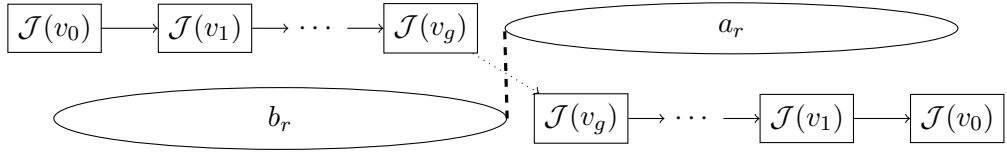

The second normal class of instances of $RO2|G=K_2|R_{\max}$ was introduced in \cite{ChPya} and is based on the following
\begin{definition}\label{def:superoverloaded} A node $v$ is referred to as \emph{superoverloaded} if jobs from ${\cal J}(v)$ can be partitioned into three subsets ${\cal J}_1,{\cal J}_2,{\cal J}_3$ such that
\begin{enumerate}
	\item $\forall k\in\{1,2,3\}$ $\sum\limits_{J_j\in{\cal J}_k}d_j\leqslant\bar R-2\dist(v_0,v)$,
	\item $\forall k\ne l\in\{1,2,3\}$ $\sum\limits_{J_j\in{\cal J}_k\cup{\cal J}_l}d_j>\bar R-2\dist(v_0,v)$.
\end{enumerate}
Such a partition is referred to as {\em irreducible} one.
\end{definition}

It was proved in \cite{ChPya} that any instance of $RO2|G=K_2|R_{\max}$ containing a superoverloaded node is normal, and the optimal schedule for such an instance can be built in linear time providing that an irreducible partition is known.% That result was generalized in \cite{ChPyaToAppear} on the $RO2||R_{\max}$ problem in the following way:
%\begin{theorem}[\cite{ChPyaToAppear}]\label{theo:ChPya}
%Let $I$ be an instance of $RO2||R_{\max}$ with superoverloaded node $v$. Then each of the following conditions imply that $I$ is normal:
%\begin{enumerate}
%	\item $v=v_0$,
%	\item There exists such an optimal cyclic route $\sigma$ in $G$ that $v$ is adjacent to $v_0$ in $\sigma$.
%\end{enumerate}
%\end{theorem}
%For the convenience of the reader we provide a necessary proof for the case $G=chain$ in Section \ref{sec:chaincases}.
We provide an elementary extension of this result on the special case with $G=chain$ in Section \ref{sec:chaincases}.

Unfortunately the verification of existence of an irreducible partition is NP-complete \cite{ChPya} and therefore the problem of obtaining of such a partition is NP-hard. However, there is a description (\cite{ChPya}) of a sufficient condition for a node to be superoverloaded, together with a polynomial time procedure of obtaining of an irreducible partition.

\begin{theorem}[\cite{ChPya}]\label{theo:superoverloaded}
Let $\Delta(v)>\frac32(\bar R-2\dist(v_0,v))+d_{\max}(v)$. Then $v$ is superoverloaded, and an irreducible partition can be obtained by the following procedure.
\end{theorem}

\begin{procedure}{Procedure {\tt Partition}.}
{Let ${\cal J}(v)=\{J_1,\dots,J_k\}$.}
\indentitem{Find minimal $x>1$ such that $\sum\limits_{j=1}^xd_j>\frac12(\bar R-2\dist(v_0,v))$.\newline
Set ${\cal J}_1=\{J_1,\dots,J_x\}$ and $X=\sum\limits_{j=1}^xd_j$.}
\indentitem{Find minimal $y>x$ such that $\sum\limits_{j=x+1}^yd_j>\bar R-2\dist(v_0,v)-X$.\newline 
Set ${\cal J}_2=\{J_{x+1},\dots,J_y\}$.}
\indentitem{Set ${\cal J}_3=\{J_{y+1},\dots,J_k\}$.}
\end{procedure}

This procedure runs correctly if for each job $J_j\in{\cal J}(v)$ its length  $d_j\leqslant\frac12(\bar R-2\dist(v_0,v))$. The condition of Theorem \ref{theo:superoverloaded} implies that inequality, and also guarantees the irreducibility of the partition obtained.
In general case the procedure {\tt Partition} still may be applied if a special treatment for ``long'' jobs (with $d_j>\frac12(\bar R-2\dist(v_0,v))$, if any) is provided. We describe the following version of the procedure guaranteed to run correctly in any case.

\begin{procedure}{Procedure {\tt Partition 2.0}.}
{Let ${\cal J}(v)=\{J_1,\dots,J_k\}$.}
\setcounter{enumi}{-1}
\indentitem{If ${\cal J}(v)$ contains long jobs, rearrange the enumeration of jobs to comply with the following conditions:
\begin{enumerate}
	\item $J_1$ is a long job,
	\item $J_2$ is also a long job unless there are no more long jobs except for $J_1$.
\end{enumerate}}
\indentitem{Find minimal $x>1$ such that $\sum\limits_{j=1}^xd_j>\frac12(\bar R-2\dist(v_0,v)+d_{\max}(v))$.\newline
{\bfseries If} no such $x$ exists, set $x=k$.\newline 
Set ${\cal J}_1=\{J_1,\dots,J_x\}$ and $X=\sum\limits_{j=1}^xd_j$.}\newline
{\bfseries If} $x=k$, set ${\cal J}_2={\cal J}_3=\emptyset$ and {\bfseries STOP}.
\indentitem{Find minimal $y>x$ such that $\sum\limits_{j=x+1}^yd_j>\bar R-2\dist(v_0,v)+d_{\max}(v)-X$.\newline 
{\bfseries If} no such $y$ exists, set $y=k$.\newline
Set ${\cal J}_2=\{J_{x+1},\dots,J_y\}$.}\newline
{\bfseries If} $y=k$, set ${\cal J}_3=\emptyset$ and {\bfseries STOP}.
\indentitem{Set ${\cal J}_3=\{J_{y+1},\dots,J_k\}$.}
\end{procedure}

Clearly, each step of the procedure requires $O(n)$ time.
However, without the conditions of Theorem \ref{theo:superoverloaded} we cannot guarantee that the partition obtained by the procedure {\tt Partition 2.0} is irreducible. Note how we use this procedure in our algorithm in the next Section.

\section{Instance reduction procedure}\label{sec:reduction}
In this section we study some general properties of an instance of \(RO2||R_{\max}\) and describe the reduction procedure which helps to reduce the number of jobs and to simplify the graph structure preserving the standard lower bound \(\bar R\). One of the important properties of the procedure is its \emph{reversibility}: any feasible schedule for a reduced instance can be treated as a feasible schedule for the initial instance with the same makespan. In general case this procedure can increase the optimal makespan. However, in the next sections we prove that for our special cases of \(RO2|G=tree|R_{\max}\) the instance reduction procedure also preserves the optimum. Therefore, it can be used as a main part of an exact algorithm for solving the initial instance.

The procedure is based on two types of instance transformation: \emph{job aggregation} and \emph{terminal edge contraction}. The first one is described in detail in \cite{ChLgot}, while the second was used in \cite{Ch} for a certain generalization of the routing open shop problem. We provide all the necessary details below.

\begin{definition}Let \(I\) be an instance of the problem \(ROm||R_{\max}\) with graph \(G=\langle V;E\rangle\), and \({\cal K}\subseteq{\cal J}(I;v)\) for some \(v\in V\). Then we say that instance \(I'\) is obtained from \(I\) by \emph{aggregation} of jobs from \({\cal K}\) if
\[{\cal J}(I';v)\doteq{\cal J}(I;v)\setminus{\cal K}\cup\{J_{j_{\cal K}}\},\,\,\forall i=1,\ldots,m\,\,p_{j_{\cal K}i}(I')\doteq\sum_{J_j\in{\cal K}}p_{ji}(I),\]
\[\forall u\ne v\,\,{\cal J}(I';u)={\cal J}(I;u).\]
(Here $j_{\cal K}$ is some new job index. A job \(J_{j_{\cal K}}\) is to replace the set of jobs \({\cal K}\).)
An instance \(\tilde I\) obtained from \(I\) by a series of job aggregations is referred to as an \emph{aggregation} of \(I\).
\end{definition}

The idea behind the job aggregation is easy: to partition jobs into some number of groups, and treat each group as a new job with processing times equal to the total processing times of the jobs combined. Similar approach is used in de Werra's algorithm for the $O2||C_{\max}$ problem \cite{dW89}.

It is easy to observe that any feasible schedule for any aggregation \(\tilde I\) can be treated as a feasible schedule for the initial instance \(I\): one just needs to replace an aggregated operation with a sequence of operations of jobs from \({\cal K}\) to be processed in any order with no idle time. Therefore, \(R^*_{\max}(\tilde I)\geqslant R^*_{\max}(I)\). Also as soon as we obtained a new job \(J_{j_{\cal K}}\) in \({\cal J}(I';v)\), it is possible that \(d_{j_{\cal K}}>d_{\max}(I;v)\), so job aggregation can lead to the growth of the standard lower bound. Specifically, (\ref{eq:SLB}) implies 
\begin{equation}\label{eq:overR}\bar R(I')>\bar R(I)\text{ if and only if }d_{j_{\cal K}}>\bar R(I)-2\dist(v_0,v).\end{equation}
We use job aggregation to simplify the instance preserving the standard lower bound. Such an aggregation is referred to as a \emph{valid} one. An instance with no further legal job aggregation possible is called \emph{irreducible}.

A natural question arises, if it is possible to perform a valid job aggregation of a whole set \({\cal J}(I;v)\) for some \(v\in V\).
To answer that question, we use the following definition from \cite{ChLgot}.

\begin{definition}\label{def:overnode}A node \(v\in V\) of an instance \(I\) of the problem \(ROm||R_{\max}\) is referred to as \emph{overloaded} if
\[\Delta(I;v)>\bar R(I)-2\dist(I;v_0,v).\]
Otherwise the node is called \emph{underloaded}.
\end{definition}
Note that Definition \ref{def:superoverloaded} implies $\Delta(v)>\frac32\bigl(\bar R-2\dist(v_0,v)\bigr)$, therefore any superoverloaded node is overloaded as well.

The job aggregation of the set \({\cal J}(I;v)\) is valid if and only if the node \(v\) is underloaded. Therefore, any node containing single job is an underloaded one.

By \(\overnodes{I}\) we denote the number of overloaded nodes in an instance \(I\).
It was proved in \cite{ChLgot} that for every instance \(I\) of the \(RO2||R_{\max}\) problem \(\overnodes{I}\leqslant 1\). Further in this section we prove a more general result (Proposition \ref{prop:overloaded}).

Now let us describe the terminal edge contraction operation. 
\begin{definition}Let \(v\in V\setminus\{v_0\}\) be some terminal node in graph \(G\), containing a single job \(J_j\) in an instance \(I\) of the \(ROm||R_{\max}\) problem. Let \(e=[v,u]\in E\) be the edge incident to \(v\). By the \emph{contraction} of the edge \(e\) we understand the following instance transformation:
\[{\cal J}(I';u)\doteq{\cal J}(I;u)\cup\{J_{j}\};\,p_{ji}(I')\doteq p_{ji}(I)+2\tau(e);\,G(I')\doteq G(I)\setminus\{v\}.\]
\end{definition}
In other words, job $J_j$ is translated from $v$ to $u$, while its operations' processing times increase by $2\tau(e)$ each. After that translation node $v$ is obsolete (contains no jobs) and to be removed from $G$.

Consider an instance \(I'\) obtained from \(I\) by the contraction of edge \(e\). Any feasible schedule for \(I'\) can be treated as a feasible schedule for the initial instance \(I\). One just needs to replace a scheduled interval of a new operation \(O_{ji}\) with three consecutive intervals: traveling of the machine \(M_i\) over the edge \(e\) to the node \(v\), performing of the old operation \(O_{ji}\), and traveling back to the node \(u\). 

Note that an edge contraction increases each machine load be $2\tau(e)$ while decreasing $T^*$ by the same amount, therefore preserving the sum $\ell_{\max}+T^*$. However the length of $J_j$ increases by $2m\tau(e)$ which might lead to the growth of the standard lower bound. We want to avoid that.
Consider the two-machine case of our problem. The following definition describes the exact condition, under which an edge contraction increases $\bar R$.
\begin{definition}\label{def:overedge} Let \(v\in V\setminus\{v_0\}\) be some terminal node in graph \(G\), containing a single job \(J_j\) in instance \(I\) of the \(RO2||R_{\max}\) problem. Let \(e=[v,u]\in E\) be the edge incident to \(v\). The edge \(e\) is referred to as \emph{overloaded} if
\begin{equation}\label{eq:overedge}d_j(I)+4\tau(e)>\bar R(I)-2\dist(I;v_0,u),\end{equation}
and \emph{underloaded} otherwise.
\end{definition}
(Note that we {\em could} perform a contraction of an overloaded edge --- meaning that the edge is terminal, the respective terminal node is not the depot and contains a single job --- but this would increase the standard lower bound. In the case the edge contraction cannot be performed, the edge is neither overloaded nor underloaded.)

For any problem instance \(I\), we denote the number of overloaded edges by \(\overedges{I}\).
The following property of any instance of \(RO2||R_{\max}\) is fundamental for the procedure of instance reduction.

\begin{proposition}\label{prop:overloaded}
Let \(I\) be an instance of the problem \(RO2||R_{\max}\). Then \(\overnodes{I}+\overedges{I}\leqslant 1\).
\end{proposition}

\begin{proof} As proved in \cite{ChLgot}, any instance of \(RO2||R_{\max}\) contains at most one overloaded node, so \(\overnodes{I}\leqslant 1\). Let us prove, that \(I\) contains at most one overloaded edge. Note that (\ref{eq:SLB}) implies 
\begin{equation}
\Delta(I)=\ell_1(I)+\ell_2(I)\leqslant 2(\bar R(I)-T^*(I)).
\label{eq:totalload}
\end{equation}
Let  \(v\) and \(v'\) be two different terminal nodes with single job in each, \(J_j\) and \(J_{j'}\) respectively; \(e=[u,v]\) and \(e'=[u',v']\) be the edges, incident to \(v\) and \(v'\), respectively, and both edges are overloaded. (Note that there is a possibility that \(u=u'\).) 
Due to the metric property of distances we have
\begin{equation}
T'^*\geqslant\dist(v_0,u)+\dist(v_0,u').
\label{eq:metricprime}
\end{equation}

From (\ref{eq:overedge}) we have
\[d_j+4\tau(e)>\bar R-2\dist(v_0,u);\,\,d_{j'}+4\tau(e')>\bar R-2\dist(v_0,u'),\]
and therefore
\begin{equation}
\Delta\geqslant d_j+d_{j'}>2\bar R-2\dist(v_0,u)-2\dist(v_0,u')-4\tau(e)-4\tau(e').
\label{eq:twooveredges}
\end{equation} 
Consider a graph \(G'=G\setminus\{v,v'\}\). Let \(T'^*\) be the optimum of the TSP on \(G'\). Then 
due to the fact that edges \(e\) and \(e'\) are terminal,
\begin{equation}
T^*=T'^*+2\tau(e)+2\tau(e').
\label{eq:terminaledges}
\end{equation}
Indeed, in order to visit terminal nodes, one needs to travel twice over their respective incident edges. 
Combining (\ref{eq:twooveredges}), (\ref{eq:metricprime}) and (\ref{eq:terminaledges}), we obtain the inequality
\[\Delta>2\bar R-2T'^*-4\tau(e)-4\tau(e')\geqslant 2\bar R-2T^*.\]
By contradiction with (\ref{eq:totalload}) we have \(\overedges{I}\leqslant 1\).

Now suppose \(\overnodes{I}=\overedges{I}=1\). Let \(e=[u,v]\) be the overloaded edge, \(v\ne v_0\) is terminal node with single job \(J_j\). Note that node \(v\) is underloaded as it contains a single job. Let \(v'\ne v\) be the overloaded node. Then, by Definitions \ref{def:overnode} and \ref{def:overedge} we have
\begin{equation}\label{eq:deltaovernode}\Delta(v')>\bar R-2\dist(v_0,v'),\end{equation}
\begin{equation}\label{eq:overedge2}d_j+4\tau(e)>\bar R-2\dist(v_0,u).\end{equation}
By using reasoning similar to that of (\ref{eq:metricprime}) and (\ref{eq:terminaledges}), we deduce
\(T^*\geqslant\dist(v_0,v')+\dist(v_0,u)+2\tau(e).\)
Using this inequality, together with (\ref{eq:deltaovernode}) and (\ref{eq:overedge2}), we obtain
\[\Delta\geqslant\Delta(v')+d_j>2\bar R-2\dist(v_0,v')-2\dist(v_0,u)-4\tau(e)\geqslant2\bar R-2T^*,\]
contradicting (\ref{eq:totalload}). This concludes the proof of the Proposition.
 \end{proof}

The idea of the following instance reduction procedure is simple. First, we aggregate jobs in all the underloaded nodes to obtain single job in each, then contract all the underloaded edges, and repeat this step until there is no underloaded edge and each underloaded node contains exactly one job. Second, we deal with the only overloaded node (if any), using the Procedure {\tt Partition 2.0} and aggregation of the obtained job sets. This way we transform the initial instance, preserving the standard lower bound. The instance obtained is irreducible. Moreover, any further operation of terminal edge contraction would increase the standard lower bound.

Note that the reduction procedure is not used for solving a problem instance, but to simplify it and consequently verify it's properties. Depending on the outcome of the procedure (see Lemma \ref{lem:reduction}) we further decide, whether the instance belong to the efficiently normal subclass we announced (see Theorem \ref{theo:scr0}). 
The simplification procedure is described in detail in Table \ref{table:reduction}. 

\begin{table}%
\begin{procedure}{\tt Reduction}{{\tt INPUT:} An instance \(I\) of the problem \(RO2||R_{\max}\).

{\tt OUTPUT:} A simplified irreducible instance \(\tilde I\).}
	\indentitem{{\bfseries For each} \emph{underloaded} \(v\in V\) perform the job aggregation of \({\cal J}(v)\).}
  \indentitem{{\bfseries For each} \emph{terminal node \(v\ne v_0\)}:
\begin{enumerate}
	\item \(e\doteq[u,v]\) (\(e\) is incident to \(v\), \(u\) is adjacent to \(v\)),
	\item {\bfseries If} \emph{\(e\) is underloaded} {\bfseries then}
	\begin{enumerate}
	  \item Let \(J_j\) is the only job in \({\cal J}(v)\),
		\item Perform the contraction of \(e\),
		\item {\bfseries If} \(u\) \emph{is underloaded} {\bfseries then} perform the job aggregation of \({\cal J}(u)\).
	\end{enumerate}
\end{enumerate}}
\indentitem{{\bfseries If} some $v\in V$ \emph{is overloaded} {\bfseries then}
\begin{enumerate}
	\item Obtain sets ${\cal J}_1,{\cal J}_2,{\cal J}_3$ by application of the {\tt Partition 2.0} to node $v$,
	\item \textbf{For each} non-empty set ${\cal J}_k$ perform the job aggregation of ${\cal J}_k$,
	\item {\bfseries If} an aggregation of two obtained jobs in ${\cal J}(v)$ with the smallest length is valid {\bfseries then} perform that aggregation.
\end{enumerate}}
\end{procedure}
\caption{The simplification procedure {\tt Reduction}.}
\label{table:reduction}
\end{table}

Let us illustrate this procedure on a small instance with transportation network from Figure \ref{fig:instance} and job data from Table \ref{table:instance}. Job data contains the description of jobs of type $J_j(a_j,b_j)$, with each node's load calculated in the last row. For this instance we have $T^*=28$ and $\ell_1=\ell_2=28$. As soon as each job length (and even node's load) is smaller than $\l_{\max}$, we have $\bar R=\ell_{\max}+T^*=56$. We can also observe that for each node $v$ we have $\Delta(v)+2\dist(v_0,v)<\bar R$, therefore all the node are underloaded. After performing Step 1 of the procedure we will obtain an instance with a single job at each node (see Table \ref{table:sampleStep1}). 

Let us perform Step 2. Consider terminal node $v_1$. The edge $[v_1,v_2]$ is underloaded as soon as $d_3+4\tau([v_1,v_2])+2\dist(v_2,v_0)=5+4+2=11<\bar R$. After the edge contraction (Step 2.2.2) we have modified job $J_3(6,3)$ translated to node $v_2$, so $\Delta(v_2)$ is now 15. The node $v_2$ is still underloaded, so we perform job aggregation (Step 2.2.3), see Table \ref{table:sampleStep21}. Further contractions of terminal edges is show in Table \ref{table:sampleStep2}. Terminal nodes are eliminated in order $v_2$, $v_3$, $v_5$, $v_7$, $v_8$, $v_6$. After the last contraction the node $v_4$ becomes overloaded, Step 3 does nothing and the procedure stops. The initial instance is reduced to a simple instance with two nodes $v_0$ and $v_4$ with three jobs.

\begin{figure}%
\begin{center}
	\begin{tikzpicture}
		\node[draw,circle,inner sep=1pt] (v1) at (0,0) {$v_1$};
		\node[draw,circle,inner sep=1pt] (v2) at (1,0) {$v_2$};
		\node[draw,circle,inner sep=1pt] (v0) at (2,0) {$v_0$};
		\node[draw,circle,inner sep=1pt] (v3) at (2.5,1) {$v_3$};
		\node[draw,circle,inner sep=1pt] (v4) at (3,0) {$v_4$};
		\node[draw,circle,inner sep=1pt] (v5) at (3.5,1) {$v_5$};
		\node[draw,circle,inner sep=1pt] (v6) at (4,0) {$v_6$};
		\node[draw,circle,inner sep=1pt] (v7) at (4.5,1) {$v_7$};
		\node[draw,circle,inner sep=1pt] (v8) at (5,0) {$v_8$};
		\draw (v1) -- (v2) node[pos=0.5,above] {1};
		\draw (v2) -- (v0) node[pos=0.5,above] {1};
		\draw (v0) -- (v3) node[pos=0.5,sloped,above] {2};
		\draw (v0) -- (v4) node[pos=0.5,above] {2};
		\draw (v4) -- (v5) node[pos=0.5,above,sloped] {3};
		\draw (v4) -- (v6) node[pos=0.5,above] {1};
		\draw (v6) -- (v7) node[pos=0.5,above,sloped] {2};
		\draw (v6) -- (v8) node[pos=0.5,above] {2};
	\end{tikzpicture}
\end{center}\caption{A sample transportation network.}%
\label{fig:instance}%
\end{figure}
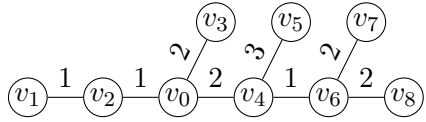

\begin{table}%
\tiny
\begin{tabular}{|c|p{7mm}|p{7mm}|p{7mm}|p{7mm}|p{7mm}|p{9mm}|p{9mm}|p{9mm}|p{9mm}|}
\hline
$v$:&$v_0$&$v_1$&$v_2$&$v_3$&$v_4$&$v_5$&$v_6$&$v_7$&$v_8$\\ \hline
Jobs:&$J_1(1,2)$ $J_2(3,4)$&$J_3(4,1)$&
$J_4(1,1)$ $J_5(1,1)$ $J_6(1,1)$&$J_7(2,1)$&
$J_8(5,2)$ $J_9(1,1)$&$J_{10}(2,1)$&$J_{11}(3,2)$&$J_{12}(1,4)$&
$J_{13}(2,3)$ $J_{14}(1,5)$\\ \hline
$\Delta(v)$:&10&5&6&3&9&3&5&5&11\\ \hline
\end{tabular}
\caption{Job data for a sample instance.}
\label{table:instance}
\end{table}

\begin{table}%
\tiny
\begin{tabular}{|c|c|c|c|c|c|c|c|c|c|}
\hline
$v$:&$v_0$&$v_1$&$v_2$&$v_3$&$v_4$&$v_5$&$v_6$&$v_7$&$v_8$\\ \hline
Jobs:&$J_{1,2}(4,6)$&$J_3(4,1)$&
$J_{4,5,6}(3,3)$&$J_7(2,1)$&
$J_{8,9}(6,3)$&$J_{10}(2,1)$&$J_{11}(3,2)$&$J_{12}(1,4)$&
$J_{13,14}(3,8)$\\ \hline
$\Delta(v)$:&10&5&6&3&9&3&5&5&11\\ \hline
\end{tabular}
\caption{Job data for a sample instance after Step 1.}
\label{table:sampleStep1}
\end{table}

\begin{table}%
\tiny
\begin{tabular}{|c|c|p{12mm}|c|c|c|c|c|c|}
\hline
$v$:&$v_0$&$v_2$&$v_3$&$v_4$&$v_5$&$v_6$&$v_7$&$v_8$\\ \hline
Jobs:&$J_{1,2}(4,6)$&
$J_3(6,3)$ $J_{4,5,6}(3,3)$&$J_7(2,1)$&
$J_{8,9}(6,3)$&$J_{10}(2,1)$&$J_{11}(3,2)$&$J_{12}(1,4)$&
$J_{13,14}(3,8)$\\ \hline
$\Delta(v)$:&10&15&3&9&3&5&5&11\\ \hline
\end{tabular}
\begin{tabular}{|c|c|c|c|c|c|c|c|c|}
\hline
$v$:&$v_0$&$v_2$&$v_3$&$v_4$&$v_5$&$v_6$&$v_7$&$v_8$\\ \hline
Jobs:&$J_{1,2}(4,6)$&
$J_{3,4,5,6}(9,6)$&$J_7(2,1)$&
$J_{8,9}(6,3)$&$J_{10}(2,1)$&$J_{11}(3,2)$&$J_{12}(1,4)$&
$J_{13,14}(3,8)$\\ \hline
$\Delta(v)$:&10&15&3&9&3&5&5&11\\ \hline
\end{tabular}
\caption{Performing Steps 2.2.2 and 2.2.3 for $v_1$.}
\label{table:sampleStep21}
\end{table}

\begin{table}%
\tiny
\begin{tabular}{|c|c|c|c|c|c|c|c|}
\hline
$v$:&$v_0$&$v_3$&$v_4$&$v_5$&$v_6$&$v_7$&$v_8$\\ \hline
Jobs:&$J_{1,2,3,4,5,6}(15,14)$&$J_7(2,1)$&
$J_{8,9}(6,3)$&$J_{10}(2,1)$&$J_{11}(3,2)$&$J_{12}(1,4)$&
$J_{13,14}(3,8)$\\ \hline
$\Delta(v)$:&29&3&9&3&5&5&11\\ \hline
\end{tabular}
\begin{tabular}{|c|c|c|c|c|c|c|}
\hline
$v$:&$v_0$&$v_4$&$v_5$&$v_6$&$v_7$&$v_8$\\ \hline
Jobs:&$J_{1,2,3,4,5,6,7}(21,19)$&
$J_{8,9}(6,3)$&$J_{10}(2,1)$&$J_{11}(3,2)$&$J_{12}(1,4)$&
$J_{13,14}(3,8)$\\ \hline
$\Delta(v)$:&40&9&3&5&5&11\\ \hline
\end{tabular}
\begin{tabular}{|c|c|c|c|c|c|}
\hline
$v$:&$v_0$&$v_4$&$v_6$&$v_7$&$v_8$\\ \hline
Jobs:&$J_{1,2,3,4,5,6,7}(21,19)$&
$J_{8,9,10}(14,10)$&$J_{11}(3,2)$&$J_{12}(1,4)$&
$J_{13,14}(3,8)$\\ \hline
$\Delta(v)$:&40&24&5&5&11\\ \hline
\end{tabular}
\begin{tabular}{|c|c|c|c|c|}
\hline
$v$:&$v_0$&$v_4$&$v_6$&$v_8$\\ \hline
Jobs:&$J_{1,2,3,4,5,6,7}(21,19)$&
$J_{8,9,10}(14,10)$&$J_{11,12}(8,10)$&
$J_{13,14}(3,8)$\\ \hline
$\Delta(v)$:&40&24&18&11\\ \hline
\end{tabular}
\begin{tabular}{|c|c|c|c|}
\hline
$v$:&$v_0$&$v_4$&$v_6$\\ \hline
Jobs:&$J_{1,2,3,4,5,6,7}(21,19)$&
$J_{8,9,10}(14,10)$&$J_{11,12,13,14}(15,22)$\\ \hline
$\Delta(v)$:&40&24&37\\ \hline
\end{tabular}
\newline
\begin{tabular}{|c|c|p{2cm}|}
\hline
$v$:&$v_0$&$v_4$\\ \hline
Jobs:&$J_{1,2,3,4,5,6,7}(21,19)$&
$J_{8,9,10}(14,10)$ $J_{11,12,13,14}(17,24)$\\ \hline
$\Delta(v)$:&40&65\\ \hline
\end{tabular}
\caption{Performing Steps 2.2.2 and 2.2.3 for nodes $v_2$, $v_3$, $v_5$, $v_7$, $v_8$, $v_6$.}
\label{table:sampleStep2}
\end{table}

Note that the procedure {\tt Reduction} obtains a reversible instance in \(O(n)\) time. Indeed, Step 1 requires $O(n)$ time. Step 2 is repeated once for each terminal node (and total number of nodes is not greater than $n$), and takes constant amount of time for a single edge contraction. Step 3 is also linear, as it's running time is majored by that of the procedure {\tt Partition 2.0}.

The following Lemma describes all possible variants of the reduced instance for the problem \(RO2|G=tree|R_{\max}\).

\begin{lemma}\label{lem:reduction} Let \(I\) be an instance of \(RO2|G=tree|R_{\max}\) and \(\tilde I\) is obtained from \(I\) by the procedure {\tt Reduction}. Then \(\bar R(\tilde I)=\bar R(I)\) and the graph \(G(\tilde I)\) satisfies exactly one of the following conditions:
\begin{enumerate}
	\item \(G(\tilde I)\) has a single node \(v_0\);
	\item \(G(\tilde I)\) is a chain connecting \(v_0\) with an overloaded node \(v\) and each node contains only one job except \(v\) which contains two or three jobs;
	\item \(G(\tilde I)\) is a chain connecting \(v_0\) with a node \(v\) with single job at each node, and the edge incident to \(v\) is overloaded.
\end{enumerate}
\end{lemma}
\begin{proof}
Each job aggregation used in the Procedure is valid and, therefore, does not grow the standard lower bound. Terminal edge contractions are applied only to underloaded edges, therefore, \(\bar R(\tilde I)=\bar R(I)\).

Consider the case \(G(\tilde I)\ne K_1\).
Note that Steps 1 and 2.2.3 guarantee that each underloaded node in \(\tilde I\) contains exactly one job. Therefore, each terminal node in \(G(\tilde I)\) is either \(v_0\), or overloaded, or incident to an overloaded edge. By Proposition \ref{prop:overloaded} 	the graph \(G(\tilde I)\) contains at most two terminal nodes (one of which is the depot) and hence is a chain. Step 2 of the procedure continues until we have no more underloaded terminal edges. Therefore, a terminal edge is contracted unless it is overloaded, or incident to the depot, or incident to an overloaded node. As soon as the first and the third options are mutually exclusive, the Lemma follows.
\end{proof}

As soon as the procedure {\tt Reduction} preserves $\bar R$, we have the following property: if the reduced instance $\tilde I$ is normal, then the initial instance $I$ is normal as well, and any normal (and hence optimal) schedule for $\tilde I$ can be easily transformed into an optimal schedule of $I$. Obviously $\tilde I$ is normal in case 1 of Lemma \ref{lem:reduction}: the problem is reduced to a classical $O2||C_{\max}$ and a normal schedule can be built by the Gonzalez-Sahni algorithm (Lemma \ref{lem:GoSa}). Therefore a class of instances of $RO2|G=tree|R_{\max}$, for which the procedure {\tt Reduction} contracts the initial tree into a single node is efficiently normal and can be solved in three steps: {\tt Reduction}, Gonzalez-Sahni algorithms and restoring a schedule for the initial instance. In the next Section we prove similar properties for case 3 and (under a certain condition) for case~2.

\section{Easy solvable cases on a chain}\label{sec:chaincases}
In this section we establish the normality of two special cases of $RO2|G=chain|R_{\max}$. In both cases we assume that the instance is irreducible, the depot is one of the ends of $G$, while the other end is either incident to an overloaded edge (which corresponds to a case 3 of Lemma \ref{lem:reduction}) or contains exactly three jobs (which is a special subcase of a case 2 of Lemma \ref{lem:reduction}). A trivial corollary of those results is Theorem \ref{theo:scr0} providing a formulation of efficiently normal subcases for the $RO2|G=tree|R_{\max}$ in terms of the outcome of the procedure {\tt Reduction} applied to an instance of the problem.

In this Section we construct early schedules using partial orders of the operations, according to the remark to the Definition \ref{def:earlyschedule}. Necessary partial orders are referred to as {\em schemes} and are described graphically. Auxiliary nodes $S$ and $F$  represent start and finish moments of a schedule.

\begin{lemma}\label{lem:overedge}Let \(I\) be an instance of \(RO2|G=chain|R_{\max}\), with \(G(I)\) being a chain \((v_0,\dots,v_g)\), \(g\geqslant 1\), each node \(v_p\) contains a single job \(J_p\) and the edge \([v_{g-1},v_g]\) is overloaded. Then one can build a normal schedule \(S\) for \(I\) in linear time.
\end{lemma}
\begin{proof}
Let \(T\doteq\dist(v_0,v_{g-1})\) and \(\mu\doteq\tau([v_{g-1},v_g])\). Then \(T^*=2(T+\mu)\).
As soon as the edge \([v_{g-1},v_g]\) is overloaded, we have
\[d_g+4\mu>\bar R-2T.\]
Therefore, (\ref{eq:totalload}) implies
\begin{equation}
\sum_{j=0}^{g-1}d_j+2T=\Delta-d_g+2T< 2\bar R-2T^*-\bar R+2T+4\mu+2T=\bar R.
\label{eq:longpathshort}
\end{equation}
Let \(S\) be the early schedule built according to the scheme from Figure \ref{fig:SchemeOverEdge}.
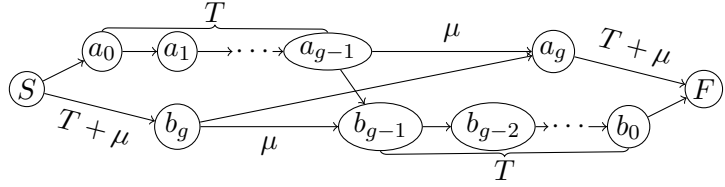
\begin{figure}%
\begin{center}
\begin{tikzpicture}
	\node[draw,circle,inner sep=1pt] (S) at (0,0) {\(S\)};
	\node[draw,circle,inner sep=1pt] (F) at (9,0) {\(F\)};
	\node[draw,circle,inner sep=1pt] (a0) at (1,0.5) {\(a_0\)};
	\node[draw,circle,inner sep=1pt] (a1) at (2,0.5) {\(a_1\)};
	\node[inner sep=1pt] (dots1) at (3,0.5) {\(\cdots\)};
	\node[draw,ellipse,inner sep=1pt] (ak1) at (4,0.5) {\(a_{g-1}\)};
	\node[draw,ellipse,inner sep=1pt] (bk1) at (4.7,-0.5) {\(b_{g-1}\)};
	\node[draw,ellipse,inner sep=1pt] (bk2) at (6.2,-0.5) {\(b_{g-2}\)};
	\node[inner sep=1pt] (dots2) at (7.2,-0.5) {\(\cdots\)};
	\node[draw,circle,inner sep=1pt] (b0) at (8,-0.5) {\(b_0\)};
	\node[draw,circle,inner sep=1pt] (bk) at (2,-0.5) {\(b_g\)};
	\node[draw,circle,inner sep=1pt] (ak) at (7,0.5) {\(a_g\)};
	\draw[->] (S) -- (a0);
	\draw[->] (a0) -- (a1);
	\draw[->] (a1) -- (dots1);
	\draw[->] (dots1) -- (ak1);
	\draw[->] (ak1) -- (ak) node[above,pos=0.5] {\(\mu\)};
	\draw[->] (ak) -- (F) node[above,pos=0.5,sloped] {\(T+\mu\)};
	\draw[->] (ak1) -- (bk1);
	\draw[->] (S) -- (bk) node[below,pos=0.5,sloped] {\(T+\mu\)};
	\draw[->] (bk) -- (bk1) node[below,pos=0.5] {\(\mu\)};
	\draw[->] (bk1) -- (bk2);
	\draw[->] (bk2) -- (dots2);
	\draw[->] (dots2) -- (b0);
	\draw[->] (b0) -- (F);
	\draw[->] (bk) -- (ak);
	
  \draw[decoration={brace},decorate] (a0.north) -- (ak1.north) node[pos=0.5,above] {\(T\)};	
	\draw[decoration={brace,mirror},decorate] (bk1.south) -- (b0.south) node[pos=0.5,below] {\(T\)};	
\end{tikzpicture}
\end{center}
\caption{A scheme of an optimal schedule for an instance with overloaded edge.}%
\label{fig:SchemeOverEdge}%
\end{figure}
Following a well-known fact from project planning, the makespan of the schedule \(S\) coincides with the length of a critical path in graph from Figure \ref{fig:SchemeOverEdge}:
\[R_{\max}(S)=\max\left\{\ell_1+T^*,\ell_2+T^*,d_g+2\dist(v_0,v_g),\sum_{j=0}^{g-1}d_j+2T\right\}.\]
From (\ref{eq:SLB}) and (\ref{eq:longpathshort}) we obtain \(R_{\max}(S)=\bar R(I),\)
concluding the proof.  \end{proof}

%The following Lemma is a corollary from Theorem \ref{theo:ChPya} \cite{ChPyaToAppear}, and the proof here is provided for the convenience of the reader only.

\begin{lemma}\label{lem:superovernode}Let \(I\) be an irreducible instance of \(RO2|G=chain|R_{\max}\), with \(G(I)\) being a chain \((v_0,\dots,v_g)\), \(g\geqslant 1\), and \(v_g\) contains three jobs, while each underloaded node \(v_p\) contains a single job \(J_p\), $p=0,\dots,g-1$. Then one can build a normal schedule \(S\) for \(I\) in linear time.
\end{lemma}

\begin{proof}
Let \({\cal J}(v_g)=\{J_{\alpha},J_{\beta},J_{\gamma}\}\). Let \(T\doteq\dist(v_0,v_g)\), then we have \(T^*=2T\).
Without loss of generality we may assume 
\begin{equation}\label{eq:assumption}a_{\alpha}\leqslant\min\{a_{\gamma},b_{\gamma},b_{\alpha}\}
%=\min\bigl\{\min\{a_k,b_k\}|k=\alpha,\beta,\gamma\bigr\}
\end{equation} (this can be achieved by renumeration of machines and/or jobs \(J_{\alpha},J_{\gamma}\)).

As soon as $I$ is irreducible we have
\begin{equation}
d_{\alpha}+d_{\beta}>\bar R-2\dist(v_0,v_g)=\bar R-T^*.
\label{eq:alphabeta}
\end{equation}

Together (\ref{eq:totalload}) and (\ref{eq:alphabeta}) imply
\begin{equation}
\sum_{j=0}^{g-1}d_j+d_{\gamma}+T^*<\bar R.
\label{eq:noncritical}
\end{equation}

Consider the early schedules \(S_1\) and \(S_2\) built according to the schemes from Figures \ref{fig:SchemeOverNode} and \ref{fig:SchemeOverNode2}, accordingly.

\begin{figure}%
\begin{center}
\begin{tikzpicture}
	\node[draw,circle,inner sep=1pt] (S) at (0,0) {\(S\)};
	\node[draw,circle,inner sep=1pt] (F) at (9.5,0) {\(F\)};
	\node[draw,circle,inner sep=1pt] (a0) at (8.5,0.5) {\(a_0\)};
	\node[inner sep=1pt] (dotsa1) at (7.7,0.5) {\(\cdots\)};
	\node[draw,ellipse,inner sep=1pt] (ak1) at (6.5,0.5) {\(a_{g-1}\)};
	\node[draw,ellipse,inner sep=1pt] (ag) at (5.3,0.5) {\(a_{\gamma}\)};
	
	\node[draw,circle,inner sep=1pt] (aa) at (3,0.5) {\(a_{\alpha}\)};
	\node[draw,circle,inner sep=1pt] (ab) at (4,0.5) {\(a_{\beta}\)};
	
	\node[draw,circle,inner sep=1pt] (ba) at (5.5,-0.5) {\(b_{\alpha}\)};
	\node[draw,circle,inner sep=1pt] (bb) at (6.5,-0.5) {\(b_{\beta}\)};
	
	\node[draw,ellipse,inner sep=1pt] (bg) at (4.3,-0.5) {\(b_{\gamma}\)};
	\node[draw,ellipse,inner sep=1pt] (bk1) at (3.2,-0.5) {\(b_{g-1}\)};
	\node[inner sep=1pt] (dotsb1) at (2,-0.5) {\(\cdots\)};
	\node[draw,circle,inner sep=1pt] (b0) at (1,-0.5) {\(b_0\)};
	\draw (S) edge[->,bend left=15] node[pos=0.5,above] {$T$} (aa);
	\draw[->] (aa) -- (ab);
	\draw[->] (ab) -- (ag);
	\draw[->] (ag) -- (ak1);
	\draw[->] (ak1) -- (dotsa1);
	\draw[->] (dotsa1) -- (a0);
	\draw[->] (a0) -- (F);
	
	\draw[->] (S) -- (b0);
	\draw[->] (b0) -- (dotsb1);
	\draw[->] (dotsb1) -- (bk1);
	\draw[->] (bk1) -- (bg);
	\draw[->] (bg) -- (ba);
	\draw[->] (ba) -- (bb);
	\draw (bb) edge[->,bend right=15] node[pos=0.5,below] {$T$} (F);
	
	\draw[<-] (ba) -- (aa);
	\draw[<-] (bb) -- (ab);
	\draw[->] (bg) -- (ag);

  \draw[decoration={brace,mirror},decorate] (a0.north) -- (ag.north) node[pos=0.5,above] {\(T\)};	
	\draw[decoration={brace},decorate] (bg.south) -- (b0.south) node[pos=0.5,below] {\(T\)};	
\end{tikzpicture}
\end{center}
\caption{A scheme of the schedule \(S_1\).}%
\label{fig:SchemeOverNode}%
\end{figure}
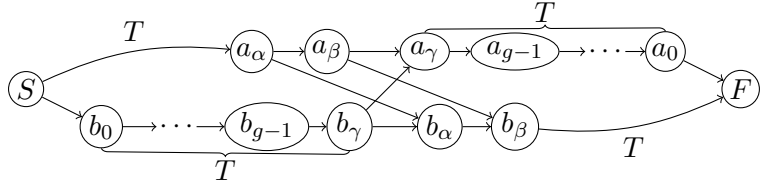
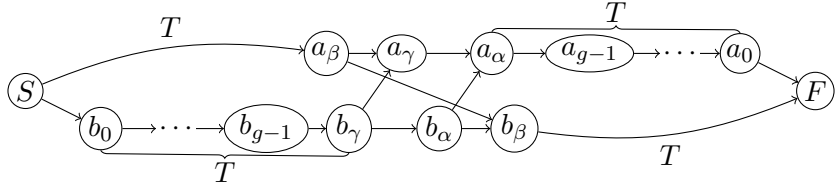
\begin{figure}%
\begin{center}
\begin{tikzpicture}
	\node[draw,circle,inner sep=1pt] (S) at (0,0) {\(S\)};
	\node[draw,circle,inner sep=1pt] (F) at (10.5,0) {\(F\)};
	\node[draw,circle,inner sep=1pt] (a0) at (9.5,0.5) {\(a_0\)};
	\node[inner sep=1pt] (dotsa1) at (8.7,0.5) {\(\cdots\)};
	\node[draw,ellipse,inner sep=1pt] (ak1) at (7.5,0.5) {\(a_{g-1}\)};
	\node[draw,ellipse,inner sep=1pt] (ag) at (5,0.5) {\(a_{\gamma}\)};
	
	\node[draw,circle,inner sep=1pt] (aa) at (6.2,0.5) {\(a_{\alpha}\)};
	\node[draw,circle,inner sep=1pt] (ab) at (4,0.5) {\(a_{\beta}\)};
	
	\node[draw,circle,inner sep=1pt] (ba) at (5.5,-0.5) {\(b_{\alpha}\)};
	\node[draw,circle,inner sep=1pt] (bb) at (6.5,-0.5) {\(b_{\beta}\)};
	
	\node[draw,ellipse,inner sep=1pt] (bg) at (4.3,-0.5) {\(b_{\gamma}\)};
	\node[draw,ellipse,inner sep=1pt] (bk1) at (3.2,-0.5) {\(b_{g-1}\)};
	\node[inner sep=1pt] (dotsb1) at (2,-0.5) {\(\cdots\)};
	\node[draw,circle,inner sep=1pt] (b0) at (1,-0.5) {\(b_0\)};
	\draw (S) edge[->,bend left=15] node[pos=0.5,above] {$T$} (ab);
	\draw[->] (ab) -- (ag);
	\draw[->] (ag) -- (aa);
	\draw[->] (aa) -- (ak1);
	\draw[->] (ak1) -- (dotsa1);
	\draw[->] (dotsa1) -- (a0);
	\draw[->] (a0) -- (F);
	
	\draw[->] (S) -- (b0);
	\draw[->] (b0) -- (dotsb1);
	\draw[->] (dotsb1) -- (bk1);
	\draw[->] (bk1) -- (bg);
	\draw[->] (bg) -- (ba);
	\draw[->] (ba) -- (bb);
	\draw (bb) edge[->,bend right=15] node[pos=0.5,below] {$T$} (F);
	
	\draw[->] (ba) -- (aa);
	\draw[<-] (bb) -- (ab);
	\draw[->] (bg) -- (ag);

  \draw[decoration={brace,mirror},decorate] (a0.north) -- (aa.north) node[pos=0.5,above] {\(T\)};	
	\draw[decoration={brace},decorate] (bg.south) -- (b0.south) node[pos=0.5,below] {\(T\)};	
\end{tikzpicture}
\end{center}
\caption{A scheme of the schedule \(S_2\)}%
\label{fig:SchemeOverNode2}%
\end{figure}

By the reasoning similar to that of the proof of Lemma \ref{lem:overedge}, using (\ref{eq:assumption}) we have
\[R_{\max}(S_1)=\max\left\{\ell_1+T^*,\ell_2+T^*,\sum_{j=0}^{g-1}d_j+d_{\gamma}+T^*,T^*+a_{\alpha}+a_{\beta}+b_{\beta}\right\}.\]

We prove that at least one of $S_1$ and $S_2$ is normal. Assume otherwise.
Then \(R_{\max}(S_1)>\bar R\) together with (\ref{eq:SLB}) and (\ref{eq:noncritical}) imply 
\begin{equation}
R_{\max}(S_1)=T^*+a_{\alpha}+a_{\beta}+b_{\beta},
\label{eq:R1}
\end{equation}
and \(R_{\max}(S_2)>\bar R\) implies 
\[R_{\max}(S_2)=T^*+\sum_{j=0}^{g-1}d_j+b_{\gamma}+\max\{a_{\gamma},b_{\alpha}\}+a_{\alpha}.\]
By the assumption \eqref{eq:assumption} we have
\begin{equation}
R_{\max}(S_2)\leqslant T^*+\sum_{j=0}^{g-1}d_j+b_{\gamma}+a_{\gamma}+b_{\alpha}.
\label{eq:R2}
\end{equation}
Therefore, by \eqref{eq:totalload}
\[R_{\max}(S_1)+R_{\max}(S_2)=2T^*+\sum_{j=0}^{g-1}d_j+d_{\alpha}+d_{\beta}+d_{\gamma}=2T^*+\Delta\leqslant 2\bar R,\]
hence both $S_1$ and $S_2$ are normal. Lemma is proved by contradiction. \end{proof}

Now we are ready to declare the main result of this Section.

\begin{theorem}\label{theo:scr0}Let \(I\) be an instance of the \(RO2|G=tree|R_{\max}\) problem, $\tilde I$ is obtained from $I$ by {\tt Reduction} procedure and one of the following conditions is true:
\begin{enumerate}
	\item \(G(\tilde I)=K_1\),
	\item \(G(\tilde I)\) contains an overloaded edge,
	\item \(G(\tilde I)\) contains an superoverloaded node.
\end{enumerate}
Then one can build a normal schedule \(S\) for \(I\) in linear time.
\end{theorem}

\begin{proof} Straightforward from Lemmas \ref{lem:reduction}, \ref{lem:overedge} and \ref{lem:superovernode}.\end{proof}

Theorem \ref{theo:scr0} can be seen as a description of efficiently normal class of instances of \(RO2|G=tree|R_{\max}\), formulated in terms of the outcome of the Procedure {\tt Reduction}. In the next Section we suggest a formulation of sufficient conditions of efficient normality in terms of the properties of the initial instance (Theorem \ref{theo:scr}).

\section{Sufficient conditions of polynomial solvability}\label{sec:conditions}

Consider an instance $I$ of the $RO2|G=tree|R_{\max}$ problem. Let us introduce some notation and definitions convenient for the description of the further results.

\begin{definition}\label{def:subtree} Let $G'=\langle V';E'\rangle$ be a subtree of $G$. We define the \emph{weight} of $G'$ as
\[W(G')\doteq \sum_{v\in V'}\Delta(v)+4\sum_{e\in E'}\tau(e).\]
\end{definition}

It is easy to observe that $W(G)$ is preserved by any operation of job aggregation and terminal edge contraction (and therefore by the Procedure {\tt Reduction}). Moreover, if during the {\tt Reduction} some subtree $G'$ is completely contracted into a node $v$, then in the instance $\tilde I$ obtained we have $\Delta(\tilde I;v)=W(G')$.

\begin{notation}
\begin{itemize}
	\item Let $v\ne v_0$ and $e=[u,v]\in E$ --- the edge, incident to $v$ in the chain connecting $v_0$ and $v$. Then $G_v$ is the connected component of $G\setminus\{e\}$ containing $v$. In other words, $G_v$ is a subtree of $G$ induced by the set of all nodes $u$ such that $v$ belongs to a chain connecting  $v_0$ and $u$. For the sake of completeness let $G_{v_0}=G$.
	\item Let $e\in E$. Then by $v(e)$ we denote the node incident to $e$ such that $e$ belongs to a chain connecting $v_0$ and $v(e)$.
	\item $B_{G'}(v)$ --- the set of all nodes of $G'$, adjacent to $v$.
\end{itemize}
\end{notation}

\begin{proposition}\label{prop:noovernode}
Let $W(G_v)\leqslant\bar R-2\dist(v_0,v)$. Then during the {\tt Reduction} procedure no node from $G_v$ becomes overloaded.
\end{proposition}

\begin{proof} It is sufficient to prove that any tree $G_v$ with such a property cannot contain an overloaded node. Assume otherwise, let some $u\in G_v$ be overloaded. Then by Definition \ref{def:overnode}
\[\Delta(u)>\bar R-2\dist(v_0,u)\geqslant \bar R-2\dist(v_0,v)\geqslant W(G_v),\]
which contradicts with Definition \ref{def:subtree}.
\end{proof}

The next Theorem describes an efficiently normal class of instances of $RO2|G=tree|R_{\max}$.

\begin{theorem}\label{theo:scr} Suppose an instance $I$ of the $RO2|G=tree|R_{\max}$ problem satisfies at least one of the following properties:
\begin{enumerate}
  \item The depot $v_0$ is overloaded.
	\item $\forall v\in B_{G(I)}(v_0)$ $W(v)\leqslant \bar R-2\tau([v_0,v])$.
	\item There exists $e\in E$ such that 
\begin{equation}\label{eq:overedgeguarantee}W\left(G_{v(e)}\right)\in\bigl(\bar R-2\dist(v_0,v)-2\tau(e),\bar R-2\dist(v_0,v)\bigr].\end{equation}
	\item There exists $v\ne v_0$ such that
	\begin{enumerate}
		\item $\forall u\in B_{G_v}(v)$ $W(G_u)\leqslant\bar R-2\dist(v_0,u)$, and
		\item $W(G_v)>\frac32\bigl(\bar R-2\dist(v_0,v)\bigr)+M$, there \[M=\max\left\{d_{\max}(v),\max_{u\in B_{G_v}(v)}\bigl(W(G_u)+4\tau([v,u])\bigr)\right\}.\]
	\end{enumerate}
\end{enumerate}
Then a normal schedule for $I$ can be built in linear time.
\end{theorem}

\begin{proof}It is sufficient to prove that such an instance $I$ satisfy the conditions of Theorem \ref{theo:scr0}.
\begin{enumerate}
  \item By Lemma \ref{lem:reduction} the only possible outcome of the {\tt Reduction} procedure is that $G$ is contracted into  $v_0$, and we have condition 1 of Theorem \ref{theo:scr0}.
	\item Let us prove that in this case the initial tree is contracted by the {\tt Reduction} procedure either into $v_0$ or into a chain containing an overloaded edge. By Proposition \ref{prop:noovernode}, in this case each of the subtrees $G_v$, $v\in B_G(v_0)$ is either contracted into $v$ or there occurs an overloaded edge. In the later case we have condition 2  of Theorem \ref{theo:scr0}, otherwise each of $G_v$ is contracted into $v$, and after that $v$ is underloaded. Any further reduction can only end up either with $v_0$ or a link with overloaded edge (if any of the edges incident to $v_0$ become overloaded). Either way we have one of the conditions 1, 2 of Theorem \ref{theo:scr0}.
	\item Let us prove that condition \eqref{eq:overedgeguarantee} guarantees that the {\tt Reduction} procedure will end up with a chain containing an overloaded edge. Indeed, by \eqref{eq:overedgeguarantee} and Proposition \ref{prop:noovernode} the tree $G_{v(e)}$ is either contracted into $v(e)$ or an overloaded edge occurs during the process and we have the claim. In the first case we obtain $\Delta(v(e))=W(G_{v(e)})$, by \eqref{eq:overedgeguarantee} $v(e)$ is underloaded, and hence all the jobs from ${\cal J}(v(e))$ are aggregated into a single job $J_x$ of length $d_x=\Delta(v(e))$ during the {\tt Reduction}. 
	Let $e=[u,v]$. From \eqref{eq:overedgeguarantee} we have
	\[d_x>\bar R-2\dist(v_0,v)-2\tau(e)=\bar R-2\dist(v_0,u)-4\tau(e),\]
	and by Definition \ref{def:overedge} the edge $e$ is overloaded.
	\item Suppose that no overloaded edge occurs during the {\tt Reduction} of subtree $v$. By condition 4.1 and Proposition \ref{prop:noovernode} $G_v$ is contracted into overloaded node $v$. The value $M$ equals the maximal job length in ${\cal J}(v)$ after the reduction prior to Step 3 of the {\tt Reduction} procedure: indeed, for each $u\in B_{G_v}(v)$ all jobs from $G_u$ are reduced into a single job of length $W(G_u)$, which is further translated into $v$ while its length becomes $W(G_u)+4\tau([v,u])$ by the contraction of edge $[v,u]$. Now by condition 4.2 set of jobs from ${\cal J}(v)$ satisfy the Theorem \ref{theo:superoverloaded}, an as soon as Step 3 of the {\tt Reduction} procedure is an application of the Procedure {\tt Partition 2.0} (see Section \ref{sec:prelim}), all the jobs from ${\cal J}(v)$ are aggregated into exactly three jobs, and $v$ is superoverloaded. Hence the reduced instance satisfies condition 3 of Theorem \ref{theo:scr0}. 
\end{enumerate}
\end{proof}

%The next Theorem is actually an elementary generalization of Theorem \ref{theo:kononovcases} for $RO2|G=K_2|R_{\max}$ on a $G=tree$ case. The proof below is basically a replication of a proof from \cite{KoFPTAS} with a few necessary adjustments and therefore cannot pose as a new original result. However, it fits perfectly on the canvas of this paper, that's why we decided to include it here.

\section{Conclusion}\label{sec:conclusion}
We described the instance reduction procedure, and proved that any instance with $G=tree$ can reduced to a chain preserving the standard lower bound. We may distinguish four outcomes of that procedure: Lemma \ref{lem:reduction} describes three, but the second one actually consists of two: with 2 and 3 jobs at the overloaded node. We cannot guarantee the normality of the initial instance only in one of that four outcomes (with 2 jobs in the overloaded node). However it would be of interest to find out, how often does such an abnormal outcome occur, and therefore, how justified the use of term ``normal'' in this context is. This might be a subject for an experimental study. 

On top of that we suggest the following directions for future investigation.

%\subsection{Investigating a problem on a graph with a small circuit rank.}

%The procedure {\tt Reduction} can be applied to an arbitrary transportation network, reducing it to a graph with at most two terminal nodes. It is easy to observe that terminal edge contraction operation preserves graph's {\em circuit rank} (which is equal to $e-v+1$ for a connected graph with $e$ edges and $v$ nodes). The next step would probably be the investigation of a problem on a graph with a small circuit rank, for example a cycle or a {\em pseudotree} (a connected graph with a single cycle). %The only polynomially solvable cases for such problems known up to date are straightforward corollaries from Theorem \ref{theo:ChPya}. 
%So we propose the following problem for future investigation.

%\smallskip

{\bfseries Problem 1.} Find new normal (or efficiently normal) classes of instances of the $RO2|G=cycle|R_{\max}$ problem.

{\bfseries Problem 2.} Determine whether $RO2|G=tree|R_{\max}$ is strongly NP-hard. 

\section*{Funding}

This research was supported by the program of fundamental scientific researches of the SB RAS No I.5.1., project No 0314-2019-0014, and by the Russian Foundation for Basic Research, projects 17-01-00170, 17-07-00513 and 18-01-00747.

\bibliographystyle{tfs}
\bibliography{Chernykh}

\begin{thebibliography}{10}
\providecommand{\MR}{\relax\unskip\space MR }
\providecommand{\url}[1]{\normalfont{#1}}
\providecommand{\urlprefix}{Available at }

\bibitem{AvBe96}
I. Averbakh and O. Berman, \emph{Routing two-machine flowshop problems on
  networks with special structure}, Transportation Science 30 (1996), pp.
  303--314.

\bibitem{AvBe99}
I. Averbakh and O. Berman, \emph{A simple heuristic for {$m$}-machine flow-shop
  and its applications in routing-scheduling problems}, Operations Research 47
  (1999), pp. 165--170.

\bibitem{AvBeCh2}
I. Averbakh, O. Berman, and I. Chernykh, \emph{A 6/5-approximation algorithm
  for the two-machine routing open-shop problem on a two-node network},
  European Journal of Operational Research 166 (2005), pp. 3--24.

\bibitem{AvBeCh1}
I. Averbakh, O. Berman, and I. Chernykh, \emph{The routing open-shop problem on
  a network: Complexity and approximation}, European Journal of Operational
  Research 173 (2006), pp. 531--539.

\bibitem{BrKn}
P. Brucker, S. Knust, T. {Edwin Cheng}, and N. Shakhlevich, \emph{{Complexity
  Results for Flow-Shop and Open-Shop Scheduling Problems with Transportation
  Delays}}, Annals of Operations Research  (2004), pp. 81--106.

\bibitem{Ch}
I. Chernykh, \emph{Routing open shop with unrelated travel times}, in
  \emph{Discrete Optimization and Operations Research --- 9th International
  Conference, {DOOR} 2016, Vladivostok, Russia, September 19-23, 2016,
  Proceedings}. 2016, pp. 272--283.

\bibitem{ChLgot}
I. Chernykh and E. Lgotina, \emph{The 2-machine routing open shop on a
  triangular transportation network}, in \emph{Discrete Optimization and
  Operations Research --- 9th International Conference, {DOOR} 2016,
  Vladivostok, Russia, September 19-23, 2016, Proceedings}. 2016, pp. 284--297.

\bibitem{ChPya}
I. Chernykh and A. Pyatkin, \emph{Refinement of the optima localization for the
  two-machine routing open shop}, in \emph{Proceedings of the 8th International
  Conference on Optimization and Applications (OPTIMA’17). Vol. 1987. CEUR
  Workshop Proceedings (1987)}. 2017, pp. 131--138.

\bibitem{ChouLin07}
S. Chou and S. Lin, \emph{Museum visitor routing problem with the ballancing of
  concurrent visitors}, Complex Systems Concurrent Engineering 6 (2007), pp.
  345--353.

\bibitem{dW89}
D. de  Werra, \emph{Graph-theoretical models for preemptive scheduling},
  Advances in Project Scheduling  (1989), pp. 171--185.
  \urlprefix\url{http://infoscience.epfl.ch/record/88562}.

\bibitem{GolPya19}
M. Golovachev and A.V. Pyatkin, \emph{Routing open shop with two nodes, unit
  processing times and equal number of jobs and machines}, in
  \emph{Mathematical Optimization Theory and Operations Research, MOTOR 2019,
  Ekaterinburg, Russia, July 8--12, 2019, Proceedings}, M. Khachay, Y.
  Kochetov, and P. Pardalos, eds., Lecture Notes in Computer Science Vol.
  11548. 2019, pp. 264--276.

\bibitem{GoSa}
T.F. Gonzalez and S. Sahni, \emph{Open shop scheduling to minimize finish
  time}, J. {ACM} 23 (1976), pp. 665--679.

\bibitem{Kl93}
U. Kleinau, \emph{Two-machine shop scheduling problems with batch processing},
  Mathematical and Computer Modelling 17 (1993), pp. 55--66.

\bibitem{KoSeCh99}
A. Kononov, S. Sevastianov, and I. Tchernykh, \emph{When difference in machine
  loads leads to efficient scheduling in open shops}, Annals of Operations
  Research 92 (1999), pp. 211--239.

\bibitem{KoFPTAS}
A. Kononov, \emph{On the routing open shop problem with two machines on a
  two-vertex network}, Journal of Applied and Industrial Mathematics 6 (2012),
  pp. 318--331.

\bibitem{Kologm}
A. Kononov, \emph{O(log m)-approximation for the routing open shop problem},
  RAIRO -- Operations Research 49 (2015), pp. 383--391.

\bibitem{TheBible}
E.L. Lawler, J.K. Lenstra, A.H.G. {Rinnooy Kan}, and G.B. Shmoys,
  \emph{Sequencing and scheduling: algorithms and complexity. Logistics of
  Production and Inventory}, Elsevier, 1993.

\bibitem{LuSoStru}
I.N. Lushchakova, A.J. Soper, and V.A. Strusevich, \emph{Transporting jobs
  through a two-machine open shop}, Naval Research Logistics 56 (2009), pp.
  1--18.

\bibitem{PinSch82}
M. Pinedo and L. Schrage, \emph{Stochastic shop scheduling: a survey}, in
  \emph{Deterministic and stochastic scheduling}, M.{\relax A.H}. Dempster,
  J.K. Lenstra, and A.{\relax H.G}. {Rinnooy Kan}, eds., NATO Advanced Study
  Institute Series Vol.~84, Springer, Dordrecht,  1982, pp. 181--196.

\bibitem{Stru}
V. Strusevich, \emph{A heuristic for the two-machine open-shop scheduling
  problem with transportation times}, Discrete Applied Mathematics 93 (1999),
  pp. 287--304.

\bibitem{BevPyat16}
R. {van Bevern} and A.V. Pyatkin, \emph{Completing partial schedules for open
  shop with unit processing times and routing}, in \emph{Proceedings of the
  11th International Computer Science Symposium in Russia (CSR’16). Lecture
  Notes in Computer Science}, Vol. 9691. 2016, pp. 73--87.

\bibitem{BevPyatSev19}
R. {van Bevern}, A.V. Pyatkin, and S.V. Sevastyanov, \emph{An algorithm with
  parameterized complexity of constructing the optimal schedule for the routing
  open shop problem with unit execution times}, Siberian Electronic
  Mathematical Reports 16 (2019), pp. 42--84.

\bibitem{SSS}
D.P. Williamson, L.A. Hall, J.A. Hoogeveen, C.A.J. Hurkens, J.K. Lenstra, S.V.
  Sevast'janov, and D.B. Shmoys, \emph{Short shop schedules}, Operations
  Research 45 (1997), pp. 288--294.

\bibitem{YuLinChou10}
V.F. Yu, S. Lin, and S. Chou, \emph{The museum visitor routing problem},
  Applied Mathematics and Computation 216 (2010), pp. 719--729.

\end{thebibliography}

\end{document}